\documentclass[11pt]{article}

\usepackage[utf8]{inputenc}
\usepackage{microtype}
\usepackage{mathtools,amssymb,mathrsfs,amsthm}   
\usepackage[table,svgnames]{xcolor}
\usepackage[breaklinks,colorlinks=true,linkcolor=black,citecolor=MidnightBlue,urlcolor=MidnightBlue]{hyperref}
\usepackage{xspace}
\usepackage{tikz}
\usepackage{fullpage}
\usepackage{floatpag}

\usepackage{algorithm}
\usepackage{algorithmicx}
\usepackage{algpseudocode}

\newtheorem{theorem}{Theorem}[section]
\newtheorem{lemma}[theorem]{Lemma}

\newtheorem{corollary}[theorem]{Corollary}
   
\newtheorem{claim}[theorem]{Claim}
\newtheorem{definition}{Definition}%
\newtheorem{remark}{Remark}

\newcommand{\INIT}{\mathsf{INIT}\xspace}
\newcommand{\IDLE}{\mathsf{IDLE}\xspace}
\newcommand{\EXPLORE}{\mathsf{EXPLORE}\xspace}
\newcommand{\DONE}{\mathsf{DONE}\xspace}
\newcommand{\ACTIVE}{\mathsf{ACTIVE}\xspace}

\newcommand{\neigh}{\mathcal{N}}
\newcommand{\T}{\mathcal{T}}

\newcommand{\Explore}{\textsc{Explore}\xspace}
\newcommand{\SetParent}{\textsc{SetParent}\xspace}
\newcommand{\MarkNoChildren}{\textsc{MarkSibling}\xspace}
\newcommand{\Relay}{\textsc{Relay}\xspace}

\newcommand{\CC}{\mathsf{CC}}

\newcommand{\FUS}{fully-utilized synchronous\xspace}

\newenvironment{theorem-repeat}[1]{\begin{trivlist}
\item[\hspace{\labelsep}{\bf\noindent Theorem \ref{#1} }]\em }%
{\end{trivlist}}

\begin{document}

\begin{titlepage}
\title{Making Asynchronous Distributed Computations Robust to Noise}
\author{Keren Censor-Hillel\thanks{Technion, Department of Computer Science, \texttt{ckeren@cs.technion.ac.il}. Supported in part by the Israel Science Foundation (grant 1696/14) and the Binational Science Foundation (grant 2015803).}
\and Ran Gelles\thanks{Faculty of Engineering, Bar-Ilan University, \texttt{ran.gelles@biu.ac.il}.}
\and Bernhard Haeupler\thanks{Carnegie Mellon University,  {\texttt{haeupler@cs.cmu.edu}}. Supported in part by NSF grants CCF-1527110 and CCF-1618280.}
}

\date{}

\maketitle

\begin{abstract}
We consider the problem of making distributed computations robust to noise, in particular to worst-case (adversarial) corruptions of messages. We give a general distributed interactive coding scheme which simulates any asynchronous distributed protocol while tolerating an optimal corruption of a $\Theta(1/n)$ fraction of all messages while incurring a moderate blowup of~$O(n\log^2 n)$ in the communication complexity.

\smallskip

Our result is the first \emph{fully distributed} interactive coding scheme in which the topology of the communication network is not known in advance. 
Prior work required either a coordinating node to be connected to all other nodes in the network or assumed a synchronous network in which all nodes already know the complete topology of the network.
\end{abstract}

\thispagestyle{empty}
\end{titlepage}

\section{Introduction}

Fault tolerance is one of the central challenges in the design of distributed algorithms. Typically, computation is performed by $n$ nodes, of which some subset may be \emph{faulty} and not behave as expected. This includes \emph{crash} or \emph{Byzantine} failures. Faults can also occur as communication errors, if links suffer from, e.g., \emph{omissions}, \emph{alterations} or \emph{Byzantine} errors (see, e.g.,~\cite{Lynch1996, AttiyaW2004}).

We focus on alteration errors, in which the content of sent messages may be corrupted. Previous work in the setting of faulty channels provides fault-tolerant algorithms for several specific tasks, such as the leader election or the consensus  problem (e.g.,~\cite{Sayeed1995,Singh96,gong1998byzantine,SCY98}).

In this paper, we provide a general technique that takes an asynchronous distributed protocol as an input and outputs a simulation of this protocol that is resilient to noise. Specifically, we develop several tools whose combination allows us to obtain the first \emph{fully distributed} interactive coding scheme.

\paragraph{The Challenge.} 
Once communication is unreliable, even the simplest distributed tasks, such as flooding information over the network or constructing a BFS tree, %
become tremendously difficult to execute correctly. 
For instance, the asynchronous distributed Dijkstra or Bellman-Ford algorithms~\cite{peleg00} miserably fail when messages may be corrupted. To see why, recall that in the Bellman-Ford algorithm, each node sends to all of its neighbors its distance from the root. A node then sets its neighbor that is closest to the root as its parent. However, if messages are incorrect, the distance mechanism may fail and nodes may set their parents in an arbitrary way.

\paragraph{Our Contribution.} 
In any attempt to tolerate message corruptions, naturally, some bound on the noise must be given. Indeed, if a majority of the sent messages are corrupted, there is no hope to complete a computation correctly. On the other hand, when the noise falls below a certain threshold, fault tolerant computation can be obtained, for example, by employing various coding techniques.

The field of \emph{coding for interactive communication} (see, e.g., the survey of~\cite{gelles15}) considers the case where two or more parties carry some computation by sending messages to one another over noisy channels
and strives to devise \emph{coding schemes} with good guarantees. 
A coding scheme is a method that is given as an input a protocol~$\pi$ that assumes reliable channels, and outputs a noise-resilient protocol~$\Pi$ that simulates~$\pi$.
The two main measures upon which a coding scheme is evaluated are its \emph{noise resilience}---the fraction of noise that the resilient simulation~$\Pi$ can withstand---and its \emph{overhead}---the amount of redundancy~$\Pi$  adds in order to tolerate faults. For networks with $n$ nodes, it is easy to show that the maximal noise fraction that any resilient protocol can cope with is $\Theta(1/n)$~\cite{JKL15}. Indeed, if more than $(1/n)$-fraction of the messages are corrupted, then the noise can completely corrupt all the communication of the node that sends the least number of messages. 
The overhead depends on the network topology, communication model, and noise resilience, as we elaborate upon in Section~\ref{subsec:related_work}.

Our main result, informally stated as follows, is a deterministic  coding scheme that fortifies any asynchronous protocol designed for a noise-free setting over any network topology, such that its resilient simulation withstands the maximal $\Theta(1/n)$-fraction of noise.
\begin{theorem}\label{thm:main}
There exists a deterministic coding scheme that takes as an input any asynchronous distributed protocol~$\pi$ designed for reliable channels, and outputs an asynchronous distributed protocol~$\Pi$ that simulates~$\pi$,  is resilient to a fraction $\Theta(1/n)$ of adversarially corrupted messages, and has a multiplicative communication overhead of~$O(n\log^2 n)$.
\end{theorem}

\subsection{Techniques}

\paragraph{A Content-Oblivious BFS Construction.}
A key ingredient in our coding scheme is a BFS construction which is \emph{content oblivious}. That is, in our BFS construction, the nodes send messages to each other and \emph{ignore their content}, basing their decisions only on the order of received messages. The challenge is to be able to do this despite asynchrony and despite lack of FIFO assumptions. 
In a sense, our construction can be seen as a variant of the distributed Dijkstra algorithm, with the property that the nodes send ``empty messages'' that contain no information (alternatively, the nodes ignore the content of received messages).

Recall that the distributed Dijkstra algorithm, see, e.g., \cite[Chapter~5]{peleg00}, is initiated by some node~$r$, which governs the BFS construction layer by layer, where the construction of each layer is called a \emph{phase}. 
The invariant is that after the $p$-th phase, the algorithm has constructed a BFS tree~$T_p$ of depth $p$ rooted at $r$, where all nodes in~$T_p$ know their parent and children in~$T_p$. The base case is~$T_0=\{r\}$, and 
the construction of the first layer is as follows. 
The node $r$ sends an \textsc{Explore} message to all its neighbors, who in turn set $r$ as their parent. Each  \textsc{Explore} message is replied to with an  \textsc{Ack} message. Once $r$ receives  \textsc{Ack} messages from all of its neighbors, the first phase ends and the construction of the second layer begins. Note that $T_1$ indeed holds $r$ and all of its neighbors.

For the $p$-th phase, the root floods a message \textsc{Phase} through~$T_{p-1}$. Once a leaf in~$T_{p-1}$ receives a  \textsc{Phase} message, it sends  \textsc{Explore} to all of its neighbors, who in turn set their parent unless already in~$T_{p-1}$. Each node that receives an \textsc{Explore} replies with an \textsc{Ack} and an indication of its parent node, so that the exploring node learns which of its neighbors is a child and which is a sibling. Upon receiving an \textsc{Ack} from all of its neighbors, the node sends an \textsc{Ack} to its parent, which propagates it all the way to~$r$. Once $r$ has received \textsc{Ack} messages from all of its children, the phase is complete.

Our content-oblivious BFS construction imitates the above behavior while using only a ``single'' type of message, instead of \textsc{Phase}, \textsc{Explore} and \textsc{Ack} messages.
Specifically, the construction begins with~$r$ sending a message (\textsc{Explore}\footnote{To ease the readability, we write in parenthesis the functionality of each sent message, but we emphasize that messages in our construction contain no content at all, and the labels of \textsc{Explore} and \textsc{Ack} are given only for the analysis.}) to all of its neighbors, who in turn set $r$ as their parent and reply with a message (\textsc{Ack}). When $r$ receives a message from all of its neighbors, the first phase is complete. Then, $r$ begins the second phase by sending another message (\textsc{Explore/Phase}) to all of its neighbors. This message causes a node that has already set its parent to behave like~$r$---it sends a message to all of its neighbors (\textsc{Explore}) \emph{except for its parent}. After receiving a message (\textsc{Ack}) from all of its neighbors, it sends a message (\textsc{Ack}) to its parent.

One can easily verify that this approach behaves similarly to the Dijkstra algorithm described above, in the sense that every node sets its parent correctly. The only difference is when a node $u$ sends an (\textsc{Explore}) message to its sibling~$w$. In the Dijkstra algorithm the sibling $w$ replies by telling the exploring node $u$ that they are siblings (by indicating the parent of $w$, which is not~$u$). However, in our case messages contain no content
and $u$ is unable to distinguish whether $w$ is a child or a sibling, since in both cases $w$ should reply to the \Explore message in the same way.

Our insight is that serializing each phase provides a solution to the above ambiguity. That is, we let $r$ send a message (\textsc{Explore/Phase}) to one child at a time, waiting to receive a message (\textsc{Ack}) from that child before sending a message (\textsc{Explore/Phase}) to the next child. This gives that if a node is expecting a message (\textsc{Explore}) from its parent but instead it receives a message (\textsc{Explore}) from a non-parent neighbor, then it knows that this neighbor must be a sibling. Hence, the node can mark all siblings and distinguish them from its children.

The main advantage of not basing our construction on the content of received messages
is that the obtained BFS construction
is \emph{inherently tolerant against message corruptions}: the noise has no effect on the construction since the content of the communicated messages is already being ignored. Notice that in our construction, the nodes do not learn their distance from~$r$, in contrast to what can easily be obtained in the noise-free case. However, this will suffice for our usage of the BFS tree in our coding scheme.

\paragraph{Interactive Coding over Sparse Subgraphs.} 
A crucial framework we rely on in our simulation is a multiparty coding scheme for interactive communication by Hoza and Schulman~\cite{HS16}, which is in turn based on ideas from~\cite{RS94}. 
This coding scheme allows simulating protocols over any graph $G=(V,E)$ and withstands an $\Theta(1/|E|)$-fraction of adversarial message corruption, while incurring a \emph{constant} communication overhead. The caveat of using this scheme for our simulation is that it applies only for \emph{synchronous} protocols that communicate over $G$ in a manner which we call \emph{\FUS}: in each round, every node communicates one symbol over to \emph{each} of its neighbors. 

In order to obtain our coding scheme for asynchronous protocol with resilience $\Theta(1/n)$, we first convert the asynchronous input protocol~$\pi$ into a \FUS protocol defined over some subgraph $G'=(V,E')$ of $G$ with $|E'|=\Theta(n)$.
To this end, we use the BFS tree constructed by our content-oblivious method described above. Once we obtain a BFS tree~$\T$, we simulate each message communicated by $\pi$ via $n$ \FUS rounds over the tree~$\T$. During each of such $n$ rounds, a message of~$\pi$ is flooded throughout~$\T$ until it reaches all the nodes and, in particular, its destination node. Note that in every round, all nodes send messages over all the edges of~$\T$. This implies a communication overhead of $O(n^2 \log n)$: we have at most $n$ rounds with  $\Theta(n)$ messages per round. The $\log n$ term stems from adding the identity of the source node and the destination node to each flooded message.\footnote{Throughout this work, all logarithms are taken to base~2.}

Using the Hoza and Schulman~\cite{HS16} coding scheme taking as an input the \FUS protocol defined over the topology~$\T$ gives a resilient simulation of~$\pi$ which withstands a maximal $\Theta(1/n)$-fraction of corrupted messages. Alas, it is a synchronous simulation, while our environment is asynchronous. Hence, to complete our simulation, we need to use a \emph{synchronizer}~\cite{awerbuch85}.

\paragraph{A Root-Triggered Synchronizer.} 
In the original error-free setting, if the input protocol to a synchronizer is guaranteed to be fully-utilized then synchronization is trivial. Each node simply attaches a round number to each of its outgoing messages and produces the outgoing messages for round $i+1$ only after receiving messages for round $i$ from all of its neighbors. The key difficulty is then for non-fully utilized synchronous input algorithms, in which a node cannot simply wait to receive a message for round $i$ from all of its neighbors, as it may be the case that some of these do not exist. 

In our setting, we guarantee that we produce a fully-utilized synchronous algorithm as an input to our synchronizer. 
However, we do not assume FIFO channels, which means that we cannot rely on the naive synchronizer, despite the promise of a \FUS protocol for an input.
Thus, we need a different solution for synchronizing the messages, and our approach is based on having a single node responsible for triggering messages of each round only after the previous round has been simulated by all nodes. To this end, our synchronizer bears similarity to the classic tree-based synchronizer of Awerbuch~\cite{awerbuch85}, with the difference that it does not incur any message overhead because it is given a \FUS input.

\paragraph{A Spanner-Based Coding Scheme.} 
We show that our coding technique described above can be further improved. Routing each message over a tree~$\T$ requires $n$ rounds in the worst case for a message to reach its destination. A more efficient solution would be to route each message through a spanning subgraph $S=(V,E_S)$ of $G$ in which the distance over~$S$ of every $(u,v)\in E$ is not too large. On the other hand, the Hoza-Schulman coding scheme on $S$ has a noise resilience of~$\Theta(1/|E_S|)$, and hence we require $|E_S|$ to be $O(n)$ in order to maintain a maximal resilience level of~$\Theta(1/n)$. Luckily, for every~$G$ there exist  sparse spanning subgraphs in which $|E_S|=O(n)$ while every two neighbors in $G$ are at distance at most $O(\log n)$ in~$S$; such subgraphs are known as $O(\log n)$-spanners~\cite{peleg00,PS89}.

Flooding a message of~$\pi$ from $u$ to $v$ can be done within $O(\log n)$ rounds, in each of which $O(|E_S|)=O(n)$ messages are sent by a \FUS simulation of $\pi$, leading to our claimed communication overhead of $O(n\log^2 n)$. Here again, the extra $\log n$ term stems from adding identifiers to each flooded message.

However, flooding information over a spanner introduces several other difficulties. For instance, in contrast to the case of a tree, it is not guaranteed anymore that each message arrives only once to its destination---indeed, multiple paths may exist between any two nodes. Furthermore, when multiple nodes send messages, the congestion may cause super-polynomial delays if a simple flooding algorithm is used. Then, due to having multiple paths with arbitrary delays, messages may arrive to their destination out of order. Since the delay is super-polynomial in the worst case, adding a counter to each message will increase the overhead by $\omega(\log n)$ and damage the global overhead. 

Instead, we provide a contention-resolution flavored technique, which consists of priority-based windows for delivering the messages. 
In more detail, a message flooding starts only at the beginning of an $O(\log n)$-round window. Multiple messages that are sent during the same window may be dropped during their flooding, yet the source always learns when its message is dropped, so it can retransmit the message in the next window. A similar approach is well-known for constructing a BFS tree when no specific root is given, but our extension of this technique is more involved, since dropped messages \emph{must be resent}.

It remains to explain how to construct the $O(\log n)$-spanner over the noisy network to begin with. For this, we use our previously described tree-based coding scheme 
to simulate a distributed spanner construction, e.g., the (noiseless) construction 
of Derbel,  Mosbah, and Zemmari~\cite{DMZ10}. 
While coding this part incurs a large overhead of~$O(n^2\log n)$,
this overhead applies only to the  part of constructing the spanner, and the global overhead of our coding scheme is still dominated by the overhead of coding the input protocol over the spanner.

\subsection{Related Work}
\label{subsec:related_work}
Performing computations over noisy channels is the heart of \emph{coding for interactive communication}, initiated by Schulman~\cite{Schulman92,schulman96}. A long line of work considers the 2-party case and obtains various coding schemes, as well as bounds on their capabilities in various settings and noise models~\cite{BR14,BE14,BKN14,GHS14,FGOS15,EGH16,KR13,Haeupler14,BGMO16,GHKRW16}. See~\cite{gelles15} for a survey on interactive coding. 

Interactive coding in the multiparty setting was first considered by Rajagopalan and Schulman~\cite{RS94} for the case of random noise, where every bit is flipped with some fixed probability.
Rajagopalan and Schulman show, for any topology~$G$, a coding scheme with an overhead of $O(\log (d+1))$, where $d$ is the maximal degree of~$G$. 
Gelles, Moitra and Sahai~\cite{GMS14} provide an efficient extension to that scheme. Alon et al.~\cite{ABEGH16} show a coding scheme with an overhead of $O(1)$ for $d$-regular graphs with degree $d=n^{\Omega(1)}$. Braverman et al.~\cite{BEGH16} demonstrate a lower bound of $\Omega(\log n)$ on the communication over a star graph. All the above works assume
 \FUS protocols, in which the protocol works in rounds and in every round all nodes communicate on all the channels connected to them. Gelles and Kalai~\cite{GK17} show that if nodes are not required to speak at every round, a lower bound of $\Omega(\log n)$ on the overhead can be proved even for coding schemes over graphs with small degree, e.g., $d=2$.
 
In the case of adversarial noise, Jain, Kalai and Lewko~\cite{JKL15} show a coding scheme that is resilient to a noise fraction of~$\Theta(1/n)$ and has an overhead of~$O(1)$ in networks which contain a star as a subgraph. Lewko and Vitercik~\cite{LV15} improve the communication balance of that scheme. Hoza and Schulman~\cite{HS16} consider \FUS protocols on arbitrary graphs and show a coding with resilience $\Theta(1/|E|)$ and overhead~$O(1)$. If the topology of the network~$G$ is known to all nodes, 
the nodes can route messages over a sparser spanning graph and decrease the number of edges used by the coding scheme. In this case, Hoza and Schulman show a coding scheme with a maximal resilience level of~$\Theta(1/n)$ and an overhead of~$O((|E|\log n) / n)$.

\medskip

Previous work in distributed settings that allow edge failures are
typically different from our setting in various aspects. Most notable are
synchrony assumptions, complete communication graphs, addressing specific
distributed tasks and assuming a bound on the number of links that may
exhibit failures. This is in contrast to our work, which addresses an
asynchronous setting with an arbitrary topology, and considers the
simulation of any distributed task.
In particular, all links may send corrupted messages, with the bound being
the number of corruptions rather than the number of faulty links.
For instance, Singh~\cite{Singh96} and Sayeed, Abu-Amara and Abu-Amara~\cite{Sayeed1995} consider the specific task of leader election and agreement for complete networks. 
Gong, Lincoln, and Rushby~\cite{gong1998byzantine},
Siu, Chin and
Yang~\cite{SCY98} and Dasgupta~\cite{Dasgupta98} 
consider agreement in complete synchronous networks with both faulty nodes and faulty links.

Pelc~\cite{Pelc92} shows that if the number of Byzantine-corrupted links is bounded by~$t$, reliable communication is achievable only over graphs whose connectivity is more than~$2t$. The same work also considers the case where each link is faulty with some probability.
In a more recent work, Feinerman, Haeupler and Korman~\cite{FeinermanHK14}
also address complete synchronous networks, and study the specific
problems of broadcast and majority consensus under random errors.

Synchronizers for unreliable settings have been studied
by Awerbuch et al.~\cite{AwerbuchPPS92}, which address a dynamic setting, and
by Harrington and Somani~\cite{HarringtonS94}, which assume faulty nodes.

\subsection{Organization}
We define basic notations, our communication and noise model as well as the notion of  noise resiliet computations (i.e., coding schemes) in Section~\ref{sec:prelim}.
In Section~\ref{sec:BFS} we describe our content-oblivious BFS construction.
A coding scheme over a spanning tree with overhead $O(n^2\log n)$ is provided in Section~\ref{sec:coding}. 
Finally, a coding scheme based on an underlying spanner with an improved overhead of~$O(n\log^2 n)$ is provided in Section~\ref{sec:coding-spanner}.

\section{Preliminaries}\label{sec:prelim}

Throughout this work we assume a network described by a graph $G=(V,E)$ with $n=|V|$ nodes and $m=|E|$ edges. Each node $u\in V$ is a party that participates in the computation and each edge $(u,v)\in E$ is a bi-directional communication channel between nodes $u$ and~$v$. 
The task of the nodes is to conduct some distributed computation given by a deterministic\footnote{While we focus here on deterministic protocols, ours result also apply to randomized Monte-Carlo protocols.} protocol~$\pi$, which consists of the algorithm each node (locally) runs. In particular, the protocol dictates to each node which messages to send to which neighbor as a function of all previous communication (and possibly the node's identity, private randomness and private input, if exists). The \emph{communication complexity} of the protocol, $\CC(\pi)$, is the maximal number of bits communicated by all nodes in any instance of~$\pi$. The \emph{message complexity} of $\pi$ is the maximal number of message sent by all nodes in any instance of~$\pi$.

We assume that the topology of $G$ is known only locally, namely, each node $v$ knows only the set $\neigh_u$ of identities of its own neighbors. However, the size of the network $n$ is known to all nodes.

\paragraph{Communication Models.}
Our protocols are for the \emph{Asynchronous} communication model defined below. In addition, we describe a different communication model named the \emph{Fully-Utilized Synchronous Model}, which is common in previous interactive coding work~\cite{RS94,HS16,ABEGH16,BEGH16}. In particular, we use coding schemes defined in the \FUS model (specifically, \cite{HS16}) as primitives for encoding our asynchronous protocols (see Lemma~\ref{lem:HS} below).
\begin{itemize}
\item \emph{Asynchronous Model}.
	In this setting, there are no timing assumptions. We assume each node is asleep until receiving a message. Once a message is received, the receiver wakes up, performs some local computation, transmits one or more messages to one or more adjacent nodes and goes back to sleep. Messages can be of any length. A protocol starts by waking up a single node $r$ of its choice.
\item \emph{The Fully-Utilized Synchronous Model}.
	Communication in this model works in synchronous \emph{rounds}, determined by a global clock. At every clock tick, every node sends one symbol (from some fixed alphabet~$\Sigma$) on each and every one of the communication links connected to it. That is, at every round  exactly $2m$~symbols are being communicated. 
\end{itemize}

\paragraph{Adversary.}
We assume an all-powerful adversary that knows the network~$G$, the protocol~$\pi$ and the private inputs of the nodes (if there are any). The adversary is able to (a) corrupt messages by changing the content of a transmitted message and (b) rush or delay the delivery of messages by an unbounded but finite amount of time. We restrict the number of messages that the adversary can corrupt, namely, we assume that the adversary can corrupt at most some fixed fraction~$\mu$ of the communicated messages. We do not restrict how a message can be corrupted and, in particular, the adversary may replace a sent message $M$ with any other message $M'$ of any length and content. However, our coding scheme will have the invariant that each message contains a single symbol (from a given alphabet~$\Sigma$), thus a message corruption will be equivalent to corrupting a single symbol. Note that the adversary is \emph{not} allowed to inject new messages or completely delete existing messages.\footnote{This type of noise, commonly called \emph{insertion and deletion} noise, is known to be more difficult to deal with in the interactive setting~\cite{BGMO16} and may be destructive for asynchronous protocols~\cite{FischerLP85}. }

\paragraph{Protocol Simulation, Resilience, and Overhead.}
A protocol~$\Pi$ is said to \emph{simulate}~$\pi$, if after the completion of~$\Pi$, each node outputs the transcript it would have seen when running~$\pi$ assuming noiseless channels. 
The protocol~$\Pi$ is \emph{resilient} to a $\mu$ fraction of noise, if $\Pi$ succeeds in simulating~$\pi$ even if an all powerful adversary completely corrupts up to a fraction~$\mu$ of the messages communicated by~$\Pi$. 
The \emph{overhead} of~$\Pi$ with respect to~$\pi$ is defined by
\(
\textit{overhead}(\Pi \mid \pi)=\CC(\Pi)/\CC(\pi).
\)

A coding scheme $\mathcal{C}:\pi \to \Pi$ converts any input protocol~$\pi$ into a resilient version $\Pi=\mathcal{C}(\pi)$. The resilience of a coding scheme is the minimal resilience of any simulation generated by the coding scheme. The (asymptotic) overhead of a coding scheme considers the maximal overhead for the worst input protocol~$\pi$ 
when $\CC(\pi)$ tends to infinity. Namely,
\[
\textit{overhead}(\mathcal{C}) =   \limsup_{c\to \infty} \max_{\substack{\pi \text{ s.t.} \\  \CC(\pi)=c}} \textit{overhead}(\mathcal{C}(\pi) \mid \pi).
\]
We are mainly interested in how the overhead scales with $n$ and $m$.

\medskip 
A famous multiparty coding scheme in the \FUS model, shown by Hoza and Schulman~\cite{HS16} (based on a previous scheme~\cite{RS94}), provides a coding scheme that simulates any noiseless \FUS protocol~$\pi$ defined over some topology~$G$ with resilience $\Theta(1/m)$ and a constant overhead~$O(1)$.
\begin{lemma}[\cite{HS16}]\label{lem:HS}
In the fully-utilized synchronous model, any $T$-round protocol~$\pi$ can be simulated by a  protocol $\Pi=\textup{HS}(\pi)$ with round complexity~$O(T)$ and communication complexity~$O(\CC(\pi))$ that is resilient to   adversarial corruption of up to an  $\Theta(1/m)$~fraction of the messages.
\end{lemma}

\section{A Distributed Content-Oblivious BFS Algorithm}
\label{sec:BFS}

In this section we show a distributed construction of a BFS tree using messages whose content can be arbitrary. We call this a \emph{content-oblivious} construction. Our algorithm can be seen as a variant of a simple distributed layered-BFS algorithm, see, e.g.,~\cite{gallager82,peleg00,tel00}.

\subsection{The BFS Algorithm: Description}
\label{sec:BFS-const}

The BFS construction is initiated by one designated node $r$ we call here the \emph{root}. 
The construction builds the tree layer by layer. First, the root sends a message to all of its neighbors.
This triggers its neighbors to set~$r$ as their parent. Each such a neighbor replies a message to~$r$ to acknowledge that it has received $r$'s message. Once $r$ has received a message from all of its neighbors, it knows that the first layer is completed, and all nodes with distance 1 have set $r$ as their parent. We call the above an \Explore step.

The root then begins a second \Explore which causes all nodes at distance $2$ to set their parent and connect to the BFS tree. Specifically, the root sends a message to each of its children and waits until all children reply a message to indicate they are done. However, in contrast to previous distributed BFS algorithms, messages are sent \emph{sequentially}---the root sends a message to its next child only after receiving the acknowledgement message from its previous child.

When a node $v$ that has already set its parent $\textit{parent}_v$ receives a message \emph{from its parent}, it acts as a root and invokes an \Explore: it sends a message to all of its neighbors excluding $\textit{parent}_v$ and waits until they all send  a message back. Only then $v$ sends a message to its parent to indicate its \Explore process has completed. It is easy to see that when the root completes its $k$-th \Explore, all nodes within distance at most $k$ have set their parent and connected to the BFS tree.

A special treatment is needed when a node $u$ receives a message from a node $v$ who is \emph{not} the parent of~$u$ during a time at which $u$ is not in the middle of an \Explore step. That is, $u$ is not expecting any messages from its neighbors, except for its parent that may trigger it to initiate another \Explore step. Recalling that messages are sent to children in a sequential manner, it is easy to verify that such a message delivery may happen only when $v$ has received a message from its own parent and is now processing its own \Explore. That is, such a message indicates that $v$ is \emph{a sibling} of $u$ in the BFS tree (namely, $v$ is not a parent nor a child of $u$ in the BFS tree).  Thus, upon receiving such a message, $u$ marks $v$ as a sibling and removes it from its list of children. To simplify the presentation, as we elaborate in Remark~\ref{remark:explore-to-sibling}, in next exploration steps $u$ will keep sending messages to $v$ as if it was one of its children.

\smallskip 

One additional property that we require from our BFS construction is that all the nodes complete the algorithm \emph{at the same time}. As explained in the introduction, we use this construction as an initial part for our coding scheme. Furthermore, recall that in order to be noise-resilient, during the BFS construction the nodes ignore the content of the messages and their entire behavior is based on whether or not a message was received. However, once this construction is complete, the nodes send and receive messages according to the coding scheme and it is crucial that a node is able to distinguish messages that belong to the BFS construction from messages of the coding scheme.  

We solve this issue by making sure that each node participates in exactly $n$ steps of \Explore. Once the node has sent the $n$-th acknowledgement to its parent, the node knows that the next message \emph{from the parent} belongs to the coding scheme rather than to the BFS construction.\footnote{Note that additional messages may arrive from a sibling node for the BFS construction but still, the next message arriving from the \emph{parent} belongs to the coding scheme rather than the BFS construction.}
To make sure that each node participates in exactly $n$ \Explore steps, regardless of its distance from~$r$, we let every node initiate one additional \Explore, which we refer to as a \emph{dummy} \Explore. Specifically, when a node completes its $(n-1)$-th \Explore, and \emph{before the node sends the acknowledgement back to its parent}, it invokes another \Explore step. Now, just by counting the messages received from the parent, every node knows whether the BFS construction has completed or not.

The pseudocode of the BFS construction is given in Algorithm~\ref{alg:BFS} and Algorithm~\ref{alg:BFS-procedures}.

\renewcommand\thealgorithm{\arabic{algorithm}(a)}   %

\begin{algorithm}[htp]
\caption{Content-oblivious BFS construction: Main Algorithm}
\label{alg:BFS}
\begin{algorithmic}[1]

\Statex
\textbf{Initialization:}
All nodes begin in the $\INIT$ state. %

\Statex

\State For node $r$ designated as root:

\algblock[Name]{Begin}{End}
\Begin
	\State $\textit{parent}_r \gets \bot$
	\State $\textit{children}_r \gets  \neigh_r$
	\State $\textit{count}_r \gets  0$ 
	\State $\textit{state}_r \gets \IDLE$
	
	\Statex 
	\While {$state_r \ne \DONE$}  \Comment Perform $n$ instances of \Call{Explore}{}
		\State $r$ invokes \Call{Explore}{} 
	\EndWhile
\End
\end{algorithmic}
\end{algorithm}

\renewcommand\thealgorithm{\the\numexpr\value{algorithm}-1(b)}   %

\begin{algorithm}[htp]
\caption{Content-oblivious BFS construction: Message Handling Procedures}
\label{alg:BFS-procedures}
\begin{algorithmic}[1]
\Statex
\Statex For every node $u$ in state $\INIT$ upon receiving a message from node $v$
\Procedure {SetParent}{}
	\State $\textit{parent}_u \gets v$
	\State $\textit{children}_u \gets  \neigh_u \setminus \{v\}$
	\State $\textit{count}_u \gets  0$ 
	\State $\textit{state}_u \gets \IDLE$
	\State send a message to $v$ \Comment{an ``ACK'' message}
	
\EndProcedure

\Statex
\Statex For every node $u$ in state $\IDLE$/$\DONE$ upon receiving a message from $v\ne \textit{parent}_u$

\Procedure {\MarkNoChildren}{}
	\State $\textit{children}_u \gets \textit{children}_u \setminus \{v\}$
	\State send a message to $v$ \Comment{an ``ACK'' message}
\EndProcedure

\Statex

\Statex For every node $u$ in state $\IDLE$ upon receiving a message from $\textit{parent}_u$
\Procedure {Explore}{}
	\State $\textit{state}_u \gets \EXPLORE$
	\State $\textit{count}_u \gets  \textit{count}_u + 1$ 
	\ForAll {$v \in \neigh_u \setminus \{\textit{parent}_u\}$} \Comment{ note: \textbf{for} is sequential}
		\State send a message to $v$ \Comment{an ``Explore'' message}
			 \label{line:ExploreChild}
		\State wait until a message is received from $v$
			\label{line:WaitForAck}
	\EndFor
	
	\Statex 
	\If {$\textit{count}_u = n-1$} \label{line:dummy}
	\Comment{Extra \emph{dummy} \Explore}
		\ForAll {$v \in \textit{children}_u$} 
		\State send a message to $v$
		\State wait until a message is received from $v$
		\EndFor \label{line:dummyEnd}
	\EndIf
	
	\Statex
	\State send a message to $\textit{parent}_u$ \Comment{an ``ACK'' message}
	\Statex
	\If  {$\textit{count}_u = n-1$}   \Comment{Change state to $\DONE$ if completed; otherwise, back to $\IDLE$}
		\State $\textit{state}_u \gets \DONE$
	\Else	
		\State $\textit{state}_u \gets \IDLE$
	\EndIf
	\label{line:ACK}
	
\EndProcedure

\end{algorithmic}
\end{algorithm}

\renewcommand\thealgorithm{\the\numexpr\value{algorithm}-1}

\subsection{The BFS Algorithm: Analysis}
\label{sec:BFS-analysis}

In this section we analyze Algorithm~1 %
and show that it satisfies the following properties.
\begin{theorem}
\label{thm:BFS}
For any input $G=(V,E)$ and node $r\in V$, Algorithm~1 %
 finds a BFS tree $\mathcal{T}$ with root $r$.
Specifically, each node knows its parent in $T$ and all of its adjacent edges that belong to $\mathcal{T}$. The algorithm communicates $O(nm)$ %
messages, where no payload is needed in any messages.  
\end{theorem}

Furthermore, we show that all nodes know that the BFS construction is complete, in the following sense.
\begin{claim}
\label{clm:order}
At the end of Algorithm~1 %
all nodes are in state $\DONE$. 
Moreover, if $r$ is in state $\DONE$ then all other nodes are in state $\DONE$ as well.
\end{claim}

\begin{proof}(\textbf{Theorem~\ref{thm:BFS}})\quad
Let $\T$ be a graph on the nodes~$V$ defined at the end of Algorithm~1 %
in the following manner: If $v=\textit{parent}_u$, then $(u,v)$ is an edge in $\T$. 
We begin by proving that $\T$ is a spanning tree. This is implied by the following claim.
\begin{claim}
At the end of the $k$-th invocation of the root's \Explore step, all the nodes that are at distance $k$ from $r$ set their parent to a node with distance $k-1$ from~$r$ and move to the state $\IDLE$, and every node of distance larger than~$k$ from~$r$ is in state~$\INIT$.
\end{claim}
\begin{proof}
We prove the claim by induction on $k$.
The base case $k=1$ follows since in the first \Explore invocation all of $r$'s children run \SetParent, setting $r$ as their parent, and switch to $\IDLE$. They send message only back to $r$, hence all other nodes remain in $\INIT$.
 
Assume that the claim holds for the $k$-th invocation and consider the $(k+1)$-th invocation of \Explore by~$r$. 
Messages propagating along the BFS tree cause all nodes of distance at most $k$ to invoke \Explore (in some order). 
This triggers a message to every node of distance $k+1$, which causes it to switch its state to $\IDLE$ and set its parent to the first node (of distance $k$) that sent it a message. Note that nodes of distance $k+1$ only communicate back to their parent and do not invoke \Explore at this time, so nodes of distance larger than $k+1$ remain in state $\INIT$. At the end of the invocation each \Explore, the invoking node switches back to state $\IDLE$. 
\end{proof}

Next, we prove that each node learns which neighbors are its children and which are not.
First note that if $v=\textit{parent}_u$ then $u$ only sends $v$  a message as a reply to a prior message received from~$v$ (i.e., an ``ACK'' message at the end of an \Explore). Therefore, whenever $v$  receives a a message from~$u$ it is in state $\EXPLORE$, and such a message can never invoke the procedure \MarkNoChildren. It follows that at the end of the algorithm $v \in \textit{children}_u$.

Next, assume $(u,v)$ is an edge in~$G$ but not in~$\T$. We show that at the end of the algorithm $v\notin \textit{children}_u$ and $u\notin \textit{children}_v$.
Let $t$ be the first time after which both $u$ and $v$ have invoked \SetParent. 
We claim that both $u$ and $v$ invoke \Explore after time $t$. This is because time $t$ is within the execution of an \Explore step invoked by $r$ and before Line~\ref{line:dummy} of that execution, and hence for every node $w\neq r$ there is a time $t_w>t$ during the execution of the loop in Lines~\ref{line:dummy}--\ref{line:dummyEnd} for $r$ in which $w$ invokes \Explore.

Finally, we note that since $(u,v)$ is an edge in $G$ but not in $\T$, then neither $u$ is an ancestor of $v$ in $\T$ nor $v$ is an ancestor of $u$ in $\T$. This implies that when $v$ invokes \Explore then $u$ is in state $\IDLE$, which causes it to invoke \MarkNoChildren and hence  $v\notin \textit{children}_u$. The proof for $u\notin \textit{children}_v$ is exactly the same.

Finally let us analyze the message complexity.
In Algorithm~1 %
each node invokes \Explore for $n$ times (see also the proof of Claim~\ref{clm:order} below), where during each \Explore it sends a message on each edge. Therefore, there are $O(n)$ messages sent on each one of the $m$ edges, which amounts to a total message complexity of $O(nm) = O(|V|\cdot |E|)$.
\end{proof}

\begin{remark} 
\label{remark:explore-to-sibling}
It is possible to reduce the message complexity by sending \Explore messages only to $\textit{children}_v$ nodes. However, this must be delayed at least one \Explore step, beyond the point in time where all the neighbors have completed their first \Explore (in order to be able to identify siblings). 
The new message complexity will be $O(|V|^2+|E|)$. For simplicity, we avoid this optimization and assume \Explore messages are sent to all non-parent nodes all the time, incurring a message complexity of $O(|V|\cdot|E|)$.
\end{remark}

We now prove Claim~\ref{clm:order}. This property is important in particular for the next section, as it suggests that there is a point in time (known by the root), when all nodes have completed their BFS algorithm. In hindsight, this allows to distinguish messages that are part of the BFS construction, whose content is ignored, from messages of the coding scheme, whose content is meaningful and must not be ignored. 

\begin{proof}(\textbf{Claim~\ref{clm:order}})
~Note that the \Explore procedure works in an DFS manner: a node replies an ACK to its parent only after all of its children reply an ACK to it. Similarly, the root completes an \Explore step after receiving an ACK from all its children, which means that they have all completed their \Explore steps. 

Note that each node invokes exactly $n$ \Explore steps due to the dummy \Explore step initiated in Line~\ref{line:dummy}. To see this, consider the same algorithm without the extra \Explore in Lines~\ref{line:dummy}--\ref{line:dummyEnd} and note that nodes at distance $k$ from the root $r$ invoke exactly $n-k$ \Explore steps. Adding this extra \Explore step at every node makes all nodes invoke \Explore exactly $n$ times. Specifically, during the $n$-th invocation of \Explore by $r$, every node with distance $1$ from $r$ invokes its $(n-1)$-th \Explore step, and then, \emph{before sending an ACK to $r$} in Line~\ref{line:ACK}, it invoke its $n$-th \Explore step. 
This then continues in an inductive manner all the way to the leaves.

Only once all of its children have sent an ACK and thus terminated the protocol and switched to~$\DONE$, a node replies with an ACK to its parent and changes its state to~$\DONE$. It follows that when the root receives an ACK for the $n$-th \Explore step from all of its children, all the nodes have terminated the protocol and switched state to~$\DONE$.
\end{proof}

\section{A Distributed Interactive Coding Scheme} 
\label{sec:coding}

In this section we show how to simulate any asynchronous protocol over a noisy network whose topology is unknown in advance.
Our main theorem for this part is the following.
\newcommand{\ThmCoding}
{
Any asynchronous protocol $\pi$ over a network~$G$ can be simulated by an asynchronous protocol~$\Pi$  resilient to an $\Theta(1/n)$-fraction of adversarial message corruption, 
and it holds that $\CC(\Pi) = O(nm\log n) + \CC(\pi) \cdot O(n^2 \log n)$.
}
\begin{theorem}
\label{thm:coding}
\ThmCoding
\end{theorem}

\subsection{Obtaining a \FUS protocol from an asynchronous input protocol $\pi$}

The first ingredient we need is a way to transform an asynchronous protocol into a \FUS protocol, in order to be able to use the Hoza-Schulman coding scheme. This transformation does not need to be robust to noise, as it is not going to be executed as is, but we will rather encode the \FUS protocol and execute the robust version. Later, we transform it back into the asynchronous setting using a synchronizer that is robust to noise.

Recall that in a \FUS protocol nodes operate in rounds, where at each round every node communicates one symbol (from some fixed alphabet~$\Sigma$) on each communication channel connected to it. We will assume the alphabet is large enough to convey all the information that our coding scheme needs. In particular, we assume each symbol contains $O(\log n)$~bits.
\begin{remark}\label{rem:fixedAlphabet}
In the following, we assume the network~$G$ is composed of channels with a fixed alphabet~$\Sigma$ of size 
$\textrm{poly}(n)$. That is, each symbol contains~$O(\log n)$ bits.  
\end{remark}
In order to avoid confusion, we will use the term ``symbols'' for messages sent by the coding scheme, while using ``messages'' to indicate the information sent by the noiseless protocol~$\pi$. %

The construction of our transformation into a \FUS protocol is given in Algorithm~\ref{alg:fus}. In this construction, each node $u$ maintains a queue of symbols that it needs to relay throughout a locally known spanning tree $\T$. 
The queue is initialized with the bits of any message that $u$ needs to send according to the input protocol~$\pi$, where each bit is encapsulated in a symbol that contains the bit value, the identity of the source (i.e., of $u$), and the identity of the destination node.
Every symbol received by~$u$ is pushed into its queue, and relayed to $u$'s neighbors in future rounds.
In particular, upon receiving the symbol $(src,dest,val)$ from a node~$w$, the node $u$ pushes the vector $(src,dest,val,w)$ to its queue. If $u$ is the destination node, it does not push the symbol into its queue; instead, $u$ collects this bit for decoding the message. 

The transformation works by having each node pop a record from its queue in each round and send the obtained triplet to all of its neighbors in $\T$ except for the node $w$ from which the message was received. If the queue is empty then an empty message is sent to all neighbors in $\T$.

\begin{algorithm}[ht]
\caption{Simulating an asynchronous protocol~$\pi$ by a \FUS protocol~$\pi'$.}
\label{alg:fus}
\begin{algorithmic}[1]

\Statex
\textbf{Initialization:}
Given is a BFS tree $\T$ rooted at $r$.
\Statex

\State In every round, for every node $u$:
\algblock[Name]{Begin}{End}
\Begin
	\For {every node $v$} 
		\State Let $M_1\cdots M_\ell$ be the bit representation of a message $M$ that $u$ has to send to $v$ in $\pi$.
		\State Push $(u,v,M_1,\bot),\cdots ,(u,v,M_\ell,\bot)$ into $queue_u$
	\EndFor

	\State $(src,dest,val,w) \gets$ pop item out of $queue_u$
	\If {$(src,dest,val,w)$ is not empty}
		\State send $(src,dest,val)$ to every $v \in \neigh_u(\T) \setminus \{w\}$ and send $\bot$ to $w$
	\Else 
		\State send $\bot$ to every $v \in \neigh_u(\T)$
	\EndIf
	
	\State For every message $(src,dest,val)$ received from $w$:
	\If {$dest\ne u$}
		\State push $(src,dest,val,w)$ into $queue_u$
	\Else 
		\State collect the bits $val$ for decoding $M$\label{line:decode}
	\EndIf
\End
\end{algorithmic}
\end{algorithm}

Note that all fragments of a message are received in order at the destination, since $\T$ has no cycles. Therefore, we can assume that the protocol sends a predefined symbol that indicates the end of the message, in order to avoid an assumption of knowledge of the message length. This ensures that Line~\ref{line:decode} is well-defined.
Our transformation guarantees the following.
\begin{lemma}
\label{lemma:fus}
Algorithm~\ref{alg:fus} creates a \FUS protocol $\pi'$ that simulates $\pi$, in the sense that all messages of $\pi$ are sent and received.
The simulation~$\pi'$ has a communication overhead of $O(n^2\log{n})$ with respect to~$\pi$, and a message complexity of~$\CC(\pi)\cdot O(n^2)$.
\end{lemma}
\begin{proof}
By construction, every node sends a symbol to all of its neighbors in each round and hence Algorithm~\ref{alg:fus} is a \FUS protocol. In addition, eventually every messages of~$\pi$ reaches its destination and hence the obtained \FUS protocol simulates~$\pi$. For the communication overhead, note that $O(\log{n})$ bits of the identities of source and destination are appended to each bit sent by~$\pi$; that is, a symbol size of $O(\log n)$ bits suffices. In addition, a delivery of a single message of~$\pi$ may require $O(n)$~rounds of relaying symbols sent along the tree~$\T$. In each such round there are $O(n)$ symbols that are sent since the obtained protocol is a \FUS protocol. This implies that $O(n^2)$ symbols are communicated per each bit of~$\pi$ and gives a total communication overhead of~$O(n^2\log{n})$.

Note that this is a worst-case analysis that assumes a single bit travels within the network at each time so that another bit is sent only after a previous bit reached its destination. If several bits are sent consecutively or if several nodes send bits simultaneously, the resulting number of messages can only decrease.
\end{proof}

\subsection{Root-triggered synchronizers}
We now describe our root-triggered synchronizer, which we use in order to execute the resilient synchronous protocol (which can be obtained by using the Hoza-Schulman coding scheme) in our asynchronous setting.
We constructed a tree-based synchronizer as in Awerbuch~\cite{awerbuch85}. The synchronizer gets as an input a \FUS protocol $\Pi'$ and outputs an equivalent asynchronous protocol $\Pi$ that simulates $\Pi'$ round by round. 

We first describe our simulation of a single round of $\Pi'$ over a \emph{tree}. Our synchronizer works as follows. The root initiates the process by sending its messages, determined by~$\Pi'$, to its children. This triggers its children to send their messages to their children, but not yet to their parent, and so forth, so that messages propagate all the way to the leaves. Once a leaf receives a message, it sends its message to its parent, and similarly, any node which receives a message from all of its children sends its message to its parent. This continues inductively all the way back to the root, which eventually receives messages from all of its children and complete the simulation of this round of~$\Pi'$.

We build upon the above idea in order to simulate a fully-utilized synchronous algorithm $\Pi'$ over an arbitrary graph $S$. That is, each node $u$ has a message $m_{uv}$ designated to each one of its neighbors~$v\in \neigh_u(S)$.%
\footnote{Later, in Section~\ref{sec:coding-spanner}, we apply our root-triggered synchronizer to an input protocol on $G$ which is fully-utilized on a spanning subgraph $S$ of $G$.}  

The pseudocode is given in Algorithm~\ref{alg:FU-sim-spanner}. We single out a node $r$, which we refer to as the \emph{initiator}, which starts by sending a message to all of its neighbors in $S$.
This triggers each neighboring node to send its messages to its neighbors, but not yet to its parent, which is now simply the neighbor from which it receives the \emph{first} message. This continues inductively, and  %
only when a node receives messages from all of its neighbors it sends its message to its parent. Eventually, the initiator receives messages from all of its neighbors and completes the simulation of the round.

\begin{algorithm}[ht]
\caption{A root-triggered synchronizer for a \FUS  protocol $\Pi'$ over a graph~$S$.}
\label{alg:FU-sim-spanner}
\begin{algorithmic}[1]

\Statex
\textbf{Initialization:}
All nodes begin in the $\INIT$ state. %
\Statex

\State For node $r$ designated as initiator:
\algblock[Name]{Begin}{End}
\Begin
	\State $\textit{state}_r \gets \ACTIVE$
	\State $\textit{parent}_r \gets \bot$
	\State $\textit{children}_r \gets  \neigh_r(S)$
	
	\Statex
	\State $r$ sends $m_{rv}$ to each node $v \in \textit{children}_r$
	\State $r$ waits to receive a message $m_{vr}$ from every node $v\in \textit{children}_r$
	\State $\textit{state}_r \gets \DONE$
\End
\Statex
\State For every node $u$, upon receiving a message from $w$ when in state $\INIT$:
\algblock[Name]{Begin}{End}
\Begin
	\State $\textit{state}_u \gets \ACTIVE$
	\State $\textit{parent}_u \gets w$
	\State $\textit{children}_u \gets  \neigh_r(S)  \setminus \{w\}$
		
	\Statex
	\State $u$ sends $m_{uv}$ to each node $v \in \textit{children}_u$
	\State $u$ waits to receive a message $m_{vu}$ from every node $v\in \textit{children}_u$
	\State $u$ sends $m_{uw}$ to $w$
	\State $\textit{state}_u \gets \DONE$
\End

\end{algorithmic}
\end{algorithm}

\goodbreak

We prove the following properties of Algorithm~\ref{alg:FU-sim-spanner}.
\begin{lemma}
\label{lem:noisy-sync}
By the end of Algorithm~\ref{alg:FU-sim-spanner} each node $u$ receives the messages $m_{vu}$ from every node $v\in \neigh_u(S)$, and all nodes are in state $\DONE$.
\end{lemma}
\begin{proof}
Let~$T$ denote the tree rooted at $r$ that is induced by the edges of $S$ that connect each node $u$ with $\textit{parent}_u$. By construction, each node $u\neq r$ sets its parent to be the first node from which it receives a messages and hence $u$ sets exactly one node as its parent in an acyclic manner, inducing the tree~$T$. 

We prove by induction on the height of the nodes with respect to $T$, that each node $u$ receives the messages $m_{vu}$ from every node $v\in \neigh_u(S)$ and then switches its state to $\DONE$. Note that every node sends its messages to all of its neighbors so that eventually all such messages arrive, and we only need to verify that the message from $u$ to $parent_u$ is eventually sent.

The base case is for the leaves of~$T$, which indeed receive messages from all of their neighbors since the only messages that get delayed are messages from nodes to their parents. Assume this holds for all nodes at height $h$, and consider a node~$u$ at height~$h+1$. Node $u$ receives messages from all of its siblings in the tree. By the induction hypothesis, every child $v$ of $u$ in~$T$ receives all of its messages and switches to state~$\DONE$. This implies that in between, node $v$ sends its message $m_{vu}$ to its parent~$u$. When this happens for all nodes $v \in children_{u}$ it is the case that $u$ receives the messages~$m_{vu}$ from every node $v\in \neigh_u(S)$ and then switches its state to~$\DONE$.
\end{proof}

By having the initiator $r$ control the simulation of each round of a simulated \FUS protocol $\Pi'$, we obtain synchronization for an arbitrary number of rounds.

\begin{corollary}
\label{cor:noisy-sync}
Multiple consecutive invocations of Algorithm~\ref{alg:FU-sim-spanner} simulate  any input \FUS protocol $\Pi'$ round by round, resulting in an asynchronous protocol $\Pi$ that uses the same number of messages.
\end{corollary}

\subsection{The Coding Scheme}

We can now complete the details of our coding scheme for asynchronous networks with unknown topology. The scheme consists of two parts. In the first part, the scheme uses the BFS construction given in Section~\ref{sec:BFS} in order to obtain a BFS tree~$\T$. Note that the nodes ignore the content of messages during this part, therefore an adversary that can only modify messages cannot disturb this part. 

In the second part, the scheme translates~$\pi$ into a \FUS protocol~$\pi'$ via $O(n)$ \FUS rounds over~$\T$. This is done using Algorithm~\ref{alg:fus}. 
The protocol $\pi'$ is still non-resilient to noise and hence is not the protocol that is executed. Instead, we add a coding layer for multiparty interactive communication, namely via the Hoza-Schulman coding scheme, whose properties are given in Lemma~\ref{lem:HS}. This results in a \FUS protocol $\Pi'$ that is resilient to noise, which we then execute through our root-triggered synchronizer to obtain the asynchronous resilient protocol $\Pi$.

The complete construction is given in Algorithm~\ref{alg:coding}. We prove its communication overhead in the following lemma, and then we prove its correctness and resilience.

\begin{algorithm}[htp]
\caption{A coding scheme~$\Pi$ for any noiseless asynchronous input protocol~$\pi$.}
\label{alg:coding}
\begin{algorithmic}[1]

\Statex
\textbf{Initialization:}
All nodes begin in the $\INIT$ state. %
\Statex

\State For node $r$ designated as initiator:
\algblock[Name]{Begin}{End}
\Begin
	\State Execute Algorithm~1 %
with $r$ designated as root. Let $\T$ be the obtained BFS tree.

	\State Let $\pi'$ be a \FUS algorithm induced by $\pi$ using Algorithm~\ref{alg:fus}.
	\State Let $\Pi'= \text{HS}(\pi')$ be the Hoza-Schulman coding scheme  for~$\pi'$. %
	\State Simulate $\Pi'$ using the synchronizer of Algorithm~\ref{alg:FU-sim-spanner} over $\T$ with $r$ as the initiator.
\End
\end{algorithmic}
\end{algorithm}

\begin{lemma}
\label{lemma:coding-communication}
For any asynchronous protocol $\pi$ the coding $\Pi$ of Algorithm~\ref{alg:coding} has a communication complexity of 
\[
\CC(\Pi) = O(nm \log n) + \CC(\pi) \cdot O(n^2 \log n).
\]
\end{lemma}
\begin{proof}
Recall that we assume channels with a fixed alphabet size, so that each symbol contains $O(\log n)$~bits (Remark~\ref{rem:fixedAlphabet}). 

The $O(nm \log n)$ term follows from Theorem~\ref{thm:BFS}. The transformation of Algorithm~\ref{alg:fus} induces a communication overhead factor of $O(n^2\log{n})$ per bit of $\pi$, as shown in Lemma~\ref{lemma:fus}.

By Lemma~\ref{lem:HS} 
there exists a resilient \FUS protocol~$\Pi'$ that simulates~$\pi'$ whose message/communication complexity is linear in the message complexity of~$\pi'$. 
Finally, Corollary~\ref{cor:noisy-sync} gives that the asynchronous simulation of $\Pi'$ via Algorithm~\ref{alg:fus} has the same message and communication complexity as~$\Pi'$.

It follows that the total overhead in communication of Algorithm~\ref{alg:coding} is $O(n^2\log n)$, as claimed.
 \end{proof}

\begin{remark}
Note  that the BFS construction (Algorithm~1) %
ignores the contents of messages sent.
Hence, if we relax the assumption of Remark~\ref{rem:fixedAlphabet}, the communication complexity can be reduced by sending empty messages (without any payload) during that step. In this case 
the message complexity of $\Pi$ remains $O(mn) + \CC(\pi)\cdot O(n^2)$ yet the communication complexity effectively reduces to $\CC(\Pi) =  \CC(\pi) \cdot O(n^2 \log n)$.
\end{remark}

\begin{remark}\label{rem:SendBlocks}
In the above, each message sent in~$\pi$
is split into single bits and a separate symbol is dedicated to each such bit. However,
instead of communicating a single bit $M_i$ in each symbol, nodes can aggregate blocks of $O(\log n)$ bits, so that the payload of each symbol is a single \emph{block} (of $\pi$'s communication) while keeping the coding's symbol size of the magnitude~$O(\log n)$. 

For some protocols, namely those which send large messages, this may result in a slight logarithmic decrease in the message complexity.
This optimization, however, will not change the asymptotic overhead in the worst case, when the protocol $\pi$ communicates a single bit at a time. 
\end{remark}

\begin{lemma}
\label{lemma:coding-correct}
For any asynchronous protocol $\pi$ the coding $\Pi$ of Algorithm~\ref{alg:coding} correctly simulates~$\pi$ even if up to $\Theta(1/n)$ of the messages are adversarially corrupted.
\end{lemma}
\begin{proof}
Correctness and resilience to noise are proved as follows. 
Theorem~\ref{thm:BFS} proves the correctness of our content-oblivious BFS construction despite noise, since the contents of the sent messages are ignored by the nodes. We emphasize that by Corollary~\ref{clm:order}, all of the nodes know when to stop ignoring the content of messages for the BFS construction and start executing that synchronizer over $\Pi'$.

Lemma~\ref{lemma:fus} proves that indeed $\pi'$ is a \FUS transformation of $\pi$. By Lemma~\ref{lem:HS}, we have that $\Pi'$ is a \FUS protocol that simulates $\pi'$ in a manner that is resilient to corrupting up to $\Theta(1/|\tilde E|)$ of the messages, where $\tilde E$ is the edges over which the protocol communicates. In our case these are the edges of the BFS tree~$\T$, and hence this step is resilient to an $\Theta(1/n)$-fraction of corruptions.

Finally,  Corollary~\ref{cor:noisy-sync} gives that $\Pi'$ is executed correctly in the asynchronous setting despite noise.

We now need to sum up the maximal number of symbols that can be corrupted and the total number of communicated symbols.
Recall that the noise resilience is the ratio between these two sums. 
Since corruption can only take place on symbols of the Hoza-Schulman coding, of which there are $\CC(\pi) \cdot O(n^2)$ many, we get that the scheme is resilient to at most $\Theta(1/n) \cdot CC(\pi) \cdot O(n^2)$ corrupted symbols. The total number of symbols communicated in the scheme includes also the $O(mn)$ symbols required for constructing the BFS tree, implying that our scheme is resilient to a fraction of symbol corruption equal to
\[
\frac{\Theta(1/n) \cdot CC(\pi) \cdot O(n^2) }{ O(nm) + CC(\pi)\cdot O(n^2) }.
\]
This is asymptotically equal to an $O(1/n)$~fraction of noise when $\CC(\pi)>n$,  $\CC(\pi)\to \infty$.
\end{proof}

Lemmas~\ref{lemma:coding-communication} and~\ref{lemma:coding-correct} directly give our main theorem for this section.
\begin{theorem-repeat}{thm:coding}
\ThmCoding
\end{theorem-repeat}

\section{A Spanner-Based Distributed Interactive Coding Scheme}
\label{sec:coding-spanner}

In this section we slightly improve the overhead obtained by the coding scheme of Theorem~\ref{thm:coding}. We demonstrate a family of coding schemes with an interesting tradeoff between their overhead and resilience. The key ingredient is replacing the underlying infrastructure of the BFS tree~$\T$ with a sparse spanning graph~$S$, where we can trade off the sparseness of the graph (i.e., the number of edges it contains, and as a consequence, the resilience of the obtained coding scheme) 
with its distance distortion (i.e., the maximal distance in $S$ for any neighboring nodes in~$G$, and as a consequence, the added overhead for routing messages through~$S$ in the coding scheme).

Assume $u$ sends $v$ a message in the input protocol~$\pi$.
The coding scheme of Algorithm~\ref{alg:coding} routes every such message via the BFS tree~$\T$. This incurs a delay in~$\Pi'$, which can be of $O(n)$ rounds: in the worst case, $u$ and $v$ which are neighbors in~$G$ may now be two leaves of~$\T$ whose distance is $n$. In fact, even if their distance in~$\T$ is smaller, the coding scheme is not aware of this fact and must propagate the message to the entire network. The only guarantee we have in this case is that the message reaches its destination after at most $n$ rounds (of the underlying \FUS protocol).

In this section we suggest a way to reduce the delay factor of $n$ by routing messages over \emph{a spanner} rather than over the tree~$\T$. 
\begin{definition}[$t$-Spanner]
A subgraph $S=(V,E_S)$ is a $t$-spanner of $G=(V,E)$ if 
for every $(u,v)\in E$ it holds that
$dist(u,v)\le t$ in~$S$.
\end{definition}
Replacing the BFS tree~$\T$ with a $t$-spanner that has $s=|E_S|$ edges ensures that a message reaches its destination after at most $t$ steps (instead of~$n$). Since the noise resilience is determined by the number of edges used by the underlying \FUS protocol, by Lemma~\ref{lem:HS}, 
we obtain a resilience of~$\Theta(1/s)$.
The main result of this section is the following.
\newcommand{\ThmCodingSpanner}
{
Let $\pi_{spanner}$ be an asynchronous  distributed algorithm for constructing a $t$-spanner $S$ with $s$ edges in a noiseless setting.
Any asynchronous protocol $\pi$ over a network~$G$ with $\CC(\pi) \gg \CC(\pi_{spanner})$  can be simulated by a noise-resilient asynchronous protocol~$\Pi$ resilient to an $\Theta(1/s)$-fraction of message corruption and it holds that  $\CC(\Pi) =   \CC(\pi) \cdot O(st\log n)$.
}
\begin{theorem}
\label{thm:coding-spanner}
\ThmCodingSpanner
\end{theorem}

Specifically, due to the existence of $O(\log n)$-spanner with $O(n)$ edges~\cite{awerbuch85,PS89} (see also~\cite[Section~16]{peleg00}), we can let $\pi_{spanner}$ be a distributed construction of a spanner with the same parameters~\cite{DMZ10} and obtain the following corollary.
\newcommand{\CorCodingSpanner}
{
Let $\pi_{spanner}$ be an asynchronous distributed algorithm for constructing a $\log{n}$-spanner with $O(n)$ edges in a noiseless setting.
Any asynchronous protocol $\pi$ over a network~$G$ with $\CC(\pi) \gg \CC(\pi_{spanner})$ can be simulated by a noise-resilient asynchronous protocol~$\Pi$  resilient to an $\Theta(1/n)$-fraction of message corruption and it holds that $\CC(\Pi) =  \CC(\pi) \cdot O(n\log^2 n)$.
}
\begin{corollary}
\label{cor:coding-spanner}
\CorCodingSpanner
\end{corollary}

There are several challenges that arise when we replace the BFS tree~$\T$ with a $t$-spanner~$S$. 
As in the case of a tree~$\T$, the nodes know neither the topology of the graph nor the shortest route between any two nodes, and the only way to propagate information to its destination is by flooding it through the network.

One difficulty stems from the fact that multiple paths may exist between any two nodes in~$S$, while only a single path exists in~$\T$. This means that we are no longer guaranteed that each message arrives to its destination only \emph{once} and that consecutive messages arrive \emph{in the correct order}. 
One possible way to overcome the above issue is to
add a serial counter to each message.
Unfortunately, a counter of $O(\log n)$ bits is not enough to avoid confusion. Consider for instance the case where 
$u$ and $v$ are at distance $t=O( \log n)$ of each other\footnote{E.g., when using an $O(\log n)$-spanner of size $s=O(n)$ towards Corollary~\ref{cor:coding-spanner}.}, yet each node in between them is connected to $\sqrt{n}$ (unique) nodes. Furthermore, assume that all nodes are currently sending messages. 
If we  flood the information through the network similar to the case of the tree~$\T$, i.e., where each node holds a queue of incoming messages which it relays one by one, then the message between $u$ and $v$ going through that route may be delayed by $(\sqrt{n})^t=n^{\Theta(\log n)}$. Such a large delay cannot be recorded by a counter of $O(\log n)$~bits. Therefore, if $u$ and $v$ are also connected via another \emph{short} path (or alternatively, a path that is not congested), then messages from the long path may be confused with messages from the short path so that the correct order of the messages could not be retrieved. 
To bypass the congestion issue, our coding scheme uses a priority mechanism for contention resolution, where messages with low priority are dropped by  congested nodes and are later resent by their originating node. A full description and a careful proof that no messages are dropped are given in Section~\ref{sec:fus-spanner}.

\smallskip

The remaining issue is how to construct the spanner despite the noisy communication. 
Luckily, this can be done using Theorem~\ref{thm:coding}, i.e., by running Algorithm~\ref{alg:coding} with an input distributed protocol~$\pi_{spanner}$ for constructing spanners. For example, we can take a deterministic synchronous\footnote{In Theorem~\ref{thm:coding-spanner} we assume $\pi_{spanner}$ is an asynchronous protocol. Note that any synchronous protocol with time complexity~$\tau$ sending  messages of size $\sigma$ can be thought of as an asynchronous protocol with message complexity $O(\tau m)$ and communication complexity $O(\tau m \sigma)$. 
}
protocol for constructing a $O(\log n)$-spanner of size $O(n)$, e.g., the construction of 
Derbel,  Mosbah, and   Zemmari~\cite{DMZ10} which 
sends messages of size~$O(\log n)$ and takes $O(n)$~rounds to complete. 
On this input, Algorithm~\ref{alg:coding}  communicates 
$O(m n^2 )$ messages in the worst case (see also Remark~\ref{rem:SendBlocks}).

Once we obtain a $t$-spanner, we proceed as in Section~\ref{sec:coding} by converting the input protocol~$\pi$ to a \FUS protocol over the spanner using Algorithm~\ref{alg:fus}, coding the \FUS protocol using the Hoza-Schulman scheme, and simulating each round of the coded algorithm using the root-triggered synchronizer over the spanner via Algorithm~\ref{alg:FU-sim-spanner}. 

Another issue which requires a careful attention
is that when we execute two coded algorithms one after the other 
we must make sure that the lengths of both parts are balanced. This happens, for example, when we first run the algorithm for obtaining the spanner~$S$ and then execute the algorithm for simulating~$\pi$ over~$S$. The reason for this is that otherwise the strong adversary can choose to attack the shorter algorithm using the larger budget of errors it has due to the longer algorithm.
For example, assume that the coded construction of the spanner communicates $x$~symbols  while the coded simulation of~$\pi$ sends $y$~symbols. If $x$ is at most~$O(y/s)$, then the adversary has enough budget to fully corrupt the first part. Requiring, say, $x=y$, makes the scheme resilient to a $\Theta(1/2s)=\Theta(1/s)$ fraction of corrupted symbols (see also footnote~\ref{footnote:balance}).

\subsection{Obtaining a spanner-based \FUS protocol }
\label{sec:fus-spanner}

A main difference between the coding scheme given in Section~\ref{sec:coding} and in the spanner-based coding scheme of this section %
is how they transform the input protocol~$\pi$ into a \FUS one. The underlying idea is similar: each node breaks the messages of~$\pi$ bits, encapsulates them in a symbol that contains also their source and destination, and floods them one by one through the underlying graph---a tree in the former case and a spanner in the latter. 

As mentioned above, the main difficulty in the case of a spanner is the existence of several paths between every two nodes, which may cause a specific symbol to arrive multiple times at the receiver. Moreover, simply flooding the information through the spanner may cause very large delays on certain paths (super-polynomial delays\footnote{Hence, appending $O(\log n)$ bits of a counter to each message does not suffice in order to guarantee the correct ordering of messages at the receiver's side.}), due to the congestion caused by other symbols that are being relayed through the network. Hence, a receiver that gets multiple copies of symbols out of order with such long delays may not be able to reconstruct $\pi$'s message from them.

Our solution breaks the simulation into windows of $2t$ rounds each, where each node attempts to send only a single symbol at each such window. In order to avoid congestion we use a priority system: when two or more symbols are received by some node, it drops all the symbols except for the one with the highest priority---the one whose sender's identity is maximal. Then, the node relays this symbol to all of its neighbors.
This procedure guarantees that at least one symbol arrives at its destination, at every $2t$-round window. Furthermore, senders whose symbols were dropped can learn this event and resend their message during the next window. If a sender receives a symbol with a higher priority, it assumes that its own symbol was dropped during that window. We prove that this mechanism may have false negatives (i.e., a sender that resends a symbol while that symbol did arrive at the destination), but it does not have false positives (if a sender does not get an indication that it needs to resend, then its symbol is guaranteed to have arrived). 

The detailed transformation into a \FUS protocol is given in 
Algorithm~\ref{alg:fus-spanner}.
\begin{algorithm}[htp]
\caption{Simulation of an asynchronous protocol~$\pi$
by a \FUS protocol~$\pi'$ using a $t$-spanner $S$.}
\label{alg:fus-spanner}
\begin{algorithmic}[1]

\algblock[Name]{Begin}{End}   %

\small

\Statex
\textbf{Input:} 
An asynchronous protocol $\pi$ over $G$;
a $t$-spanner $S$ of~$G$ (the value $t$ is known to all nodes)
\Statex

\Statex \textbf{Init:} for all nodes $u$: $\textsf{RepeatSendMsg}_u \gets \textit{false}$; $\textsf{NextMsg}_u,\textsf{RelayMemory}_u,\textsf{queue}_u \gets\emptyset$.

\Statex

\State In every round~$\mathsf{RND}$, for every node $u$:
\Begin
  \If {$\mathsf{RND} \equiv 1 \mod 2t$}
  	\State Invoke  \Call{SetNextMessage}{}
  \EndIf
  	
	\State Invoke  \Call{Relay}{}
\End

\Statex

\Procedure{SetNextMessage}{} (for node $u$)
\If {$\textsf{RepeatSendMsg}_u = \textit{false}$ or $\textsf{NextMsg}_u = \emptyset$} 
	\Statex \Comment{{Split the next message of $\pi$ into bits, each to be delivered separately}}
	\State Let $M_1\cdots M_\ell$ be the bit representation of a message $M$ that $u$ has to send to $v$ in $\pi$.
	\State Push $(u,v,M_1,1 \mod n),\cdots ,(u,v,M_\ell,\ell \mod n)$ into $\textsf{queue}_u$
		\Comment{ignore if no message to send}
	\State $\textsf{NextMsg}_u \gets \mathsf{pop}(\textsf{queue}_u)$  	
	\State $\textsf{RelayMemory}_u\gets \textsf{NextMsg}_u$
	\State $\textsf{RepeatSendMsg}_u \gets \textit{false}$
\ElsIf  {$\textsf{RepeatSendMsg}_u = \textit{true}$} 
	\Comment{ \smash{\parbox[t]{0.4\columnwidth}{Previous bit may have not reached the destination; retry with the same $\textsf{NextMsg}$.}}}

	\State $\textsf{RelayMemory}_u\gets \textsf{NextMsg}_u$
	\State $\textsf{RepeatSendMsg}_u \gets \textit{false}$
\EndIf
\EndProcedure

\Statex

\Procedure{Relay}{} (for node $u$)
	\State $(src',dest',val',i',w') \gets \textsf{RelayMemory}_u$

	\Comment{Relay highest priority message to all neighbors}

	\If {$\textsf{RelayMemory}_u \ne \emptyset$}
		\State send $(src',dest',val',i')$ to every $v \in \neigh_u(S) \setminus \{w'\}$ and send $\bot$ to~$w'$
	\Else 
		\State send $\bot$ to every $v \in \neigh_u(S)$
	\EndIf
	\State $\textsf{RelayMemory}_u \gets \emptyset$
	
	\Statex 
	\Statex
	\Comment{Receive messages from all neighbors, keep only the message with highest priority }
	\State Upon receiving a message $(src,dest,val,i)$ from $w$:
\Begin
	
	\If {$src' > u$}  
	\Comment{ \smash{\parbox[t]{0.4\columnwidth}{Whenever a message with higher priority is received,  assume own message wasn't delivered during this window of $2t$ rounds}}  }
		\State $\textsf{RepeatSendMsg}_u \gets true$  \label{alg:spanner:repeatFlag}
	\EndIf
	
	\Statex %
	\If {$dest\ne u$}
			\If  {$src > src'$ or $\textit{RelayMemory}_u=\emptyset$}
				\State $\textsf{RelayMemory}_u\gets (src,dest,val,i,w)$
			\EndIf
	\ElsIf {$dest=u$}
		\State 
		\parbox[t]{0.8\linewidth}
		{%
		collect the bits $val$ for decoding $M$\\ %
		(using the index $i$ to ignore already received bits)%
		}
	\EndIf
\End
\EndProcedure
\end{algorithmic}
\thisfloatpagestyle{empty} %
\end{algorithm}

\bigskip

\goodbreak

Towards proving the correctness of Algorithm~\ref{alg:fus-spanner},
we prove some properties of the algorithm.
\begin{lemma} \label{lem:messageProgress} 
During every non-overlapping window of $2t$ rounds of $\pi$ at least one message $(src,dest,val,i)$ is delivered to its destination.
Furthermore, The sender of this message ends this window of $2t$ rounds with $\textsf{RepeatSendMsg}_u = \textit{false}$.
\end{lemma}
\begin{proof}
Since \Relay is performed according to priority, the message with the highest priority (highest $src$) always survives and gets relayed all the way to the destination (which takes at most $t$ rounds). The sender of this message never receives a message with a higher priority, thus it never sets $\textsf{RepeatSendMsg}_u$ to $\textit{true}$ on Line~\ref{alg:spanner:repeatFlag}.
\end{proof}

\begin{lemma} \label{lem:failureIndication}
If $u$ sends a message and after $2t$ rounds it holds that $\textsf{RepeatSendMsg}_u = \textit{false}$, then the message has reached its destination.
\end{lemma}
\begin{proof}
Assume $u$ sends a message to some node~$v$. Additionally, assume that a node~$w$ with a higher priority ($w>u$) also attempts to send a message during the same $2t$-round window. Without loss of generality, we assume $w$ is the only node with priority higher than~$u$; also recall that nodes with lower priority have no effect on the delivery of $u$'s message. 

For $k\le2t$, define $B_{k}$ to  be all the messages that travel during the first $k$ rounds (of that specific $2t$-round window) on all the edges of distance at most $k$ from~$u$.
We consider several cases:
\begin{enumerate}
\item
If $u$'s message has the highest priority in $B_{2t}$ then it is clear that the statement holds. This happens whenever $dist(u,w)>4t$. 

\item
If $dist(u,w) \le 2t$ then $u$ receives $w$'s message during that window, and sets $\textsf{RepeatSendMsg}_u$ to~$true$, so in this case the required conditions of the lemma do not hold to begin with.

\item
The last case is when $B_{2t}$ contains $w$'s message, yet $dist(u,w) > 2t$. This implies that $dist(v,w)>t$ (otherwise, via the trianlge inequality, we have that $dist(u,w)\le 2t$). This in turn implies that $u$'s message is the one with the highest priority in $B_{t}$,
which means that $u$'s message is delivered to $v$ by round~$t$.
\end{enumerate}
\end{proof}

Note that the opposite direction of the statement of Lemma~\ref{lem:failureIndication} does not always hold. 
That is, if $u$ sends a message and after $2t$ rounds it holds that $\textsf{RepeatSendMsg}_u = \textit{true}$, then it is possible that the message has nevertheless reached its destination. 

Lemma~\ref{lem:messageProgress} and Lemma~\ref{lem:failureIndication} suggest that progress is made every $2t$ rounds: at least one message is being delivered and all nodes whose messages are not delivered receive an indication of this event (and retry during the next window). This leads to the correctness of the algorithm, as stated in the following lemma.

\begin{lemma}
\label{lemma:fus-spanner}
Algorithm~\ref{alg:fus-spanner} creates a fully-utilized synchronous protocol $\pi'$ that simulates $\pi$, in the sense that all messages of $\pi$ are sent and received, with a communication overhead of $O(st \log{n})$, where $s$ is the number of edges in the input $t$-spanner $S$.
\end{lemma}
\begin{proof}
By construction, every node sends a message to all of its neighbors in~$S$ in each round and hence Algorithm~\ref{alg:fus-spanner} is a fully-utilized synchronous protocol. In addition, as implied by Lemma~\ref{lem:messageProgress} and Lemma~\ref{lem:failureIndication},
eventually every messages of $\pi$ reaches its destination and hence the obtained fully-utilized synchronous protocol simulates all messages in~$\pi$. 
Note that a message may be resent and received multiple times at the destination. However a node $u$ always receives from any specific node $v$ either the next bit of the message $v$ sends, or a re-transmission of the last bit of the same message. Therefore, adding a single parity bit of the index of the transmitted bit is enough in order to avoid confusion (In Algorithm~\ref{alg:FU-sim-spanner} we actually add to each bit the information $(i \mod n)$ where $i$ is the index of the bit; yet communicating the value $(i\mod 3)$ would have sufficed).

For the communication overhead, note that $O(\log n)$ bits of the identities of source and destination and bit index are appended to each bit sent by~$\pi$, that is $\log |\Sigma|=O(\log n)$ suffices. In addition, a delivery of a single bit of~$\pi$ may require $O(t)$~rounds when sent over the spanner~$S$. In each such round $O(s)$ symbols are sent in a fully-utilized synchronous protocol. This implies $O(st)$ messages are communicated per each bit of~$\pi$ which gives a total communication overhead of~$O(st \log n)$.

Note that this is a worst-case analysis that assumes a single bit (of~$\pi$) travels within the network at each time so that another bit is sent only after a previous bit reached its destination. If several bits are sent consecutively or if several nodes send bits simultaneously, the resulting number of messages can only decrease.
\end{proof}

\goodbreak

\subsection{The spanner-based coding scheme}
\label{sec:construction-coding-spanner}
Our modified spanner-based coding scheme~$\Pi$ is given in Algorithm~\ref{alg:coding-spanner}.

\begin{algorithm}[htp]
\caption{Spanner-based coding scheme~$\Pi$ for any asynchronous noiseless protocol~$\pi$.}
\label{alg:coding-spanner}
\begin{algorithmic}[1]

\Statex
\textbf{Initialization:}
All nodes know (a bound on) $\CC(\pi)$.
$\pi_{spanner}$ is an asynchronous (noiseless) protocol for constructing a $t$-spanner. 

\Statex

\State For node $r$ designated as initiator:
\algblock[Name]{Begin}{End}
\Begin
	
	\State \parbox[t]{0.9\columnwidth}{Execute Algorithm~\ref{alg:coding} on input $\pi_{spanner}$ 
with $r$ designated as root, extending the protocol so that the total message complexity of this step is~$C_2$ defined below.   
Let $S$ be the obtained $t$-spanner. %
}  \label{alg:CodingSpanner:BFSstep}

	\State Let $\pi'$ be a \FUS algorithm induced by $\pi$ using Algorithm~\ref{alg:fus-spanner} on $S$. \label{alg:CodingSpanner:FUSstep}
	\State \parbox[t]{0.9\columnwidth}{Let $\Pi'= \text{HS}(\pi')$ be the Hoza-Schulman coding scheme  for~$\pi'$.  %
	Let $C_2$ be the message complexity of~$\Pi'$.}
	\label{alg:CodingSpanner:HSstep}
	\State \parbox[t]{0.9\columnwidth}{Simulate $\Pi'$ using the synchronizer of Algorithm~\ref{alg:FU-sim-spanner} over $S$ with $r$ as the initiator. 
	} \label{alg:CodingSpanner:simulation}
\End
\Statex
\end{algorithmic}
\end{algorithm}

As mentioned above, when a coding scheme contains two parts 
(i.e., constructing the spanner and executing the coding scheme)
that are being coded independently, it is necessary to make sure that 
the two parts  are of equal length. 
Recall that we consider the asymptotical behavior of the coding scheme when $\CC(\pi)$ tends to infinity. Hence, we assume $CC(\pi)\gg CC(\pi_{spanner})$, which implies that we need to extend the first part of our scheme where we construct the spanner graph in a resilient way.
There are two possible ways 
to make this extension: either by artificially increasing the communication of $\pi_{spanner}$ to $\approx O(\CC(\pi)\cdot st / n^2)$ bits, say, by sending zeroes after the spanner's construction has completed; or by running the Hoza-Schulman encoding for $O(C_2/ n )$ rounds.

We claim that the coding scheme~$\Pi$ of Algorithm~\ref{alg:coding-spanner} satisfies the requirements of Theorem~\ref{thm:coding-spanner}.

\begin{proof} 
(\textbf{Theorem~\ref{thm:coding-spanner}})
We first analyze the communication complexity of the scheme. We recall the assumption stated in Remark~\ref{rem:fixedAlphabet} that symbol of the coding scheme contains $O(\log n)$ bits; thus we can equivalently bound the message complexity.
The execution of Algorithm~\ref{alg:coding} in Line~\ref{alg:CodingSpanner:BFSstep} has a message complexity of $C_1=O(mn)+\CC(\pi_{spanner})\cdot O(n^2)$. However, we artificially extend this step, so that it would have a message complexity of~$C_2$.
By Lemma~\ref{lemma:fus-spanner}, the message complexity of $\pi'$ given by Algorithm~\ref{alg:fus-spanner} (Line~\ref{alg:CodingSpanner:FUSstep}) is $C_2=\CC(\pi)\cdot O(st)$. Hence, the overall message complexity of Algorithm~\ref{alg:coding-spanner} is~$2C_2$, and the overall communication complexity is $O(C_2 \log n)=\CC(\pi) \cdot O(ts\log n)$.

Regarding the resilience of the scheme,
the first part (Line~\ref{alg:CodingSpanner:BFSstep}) is resilient due to Theorem~\ref{thm:coding};
the second part (Lines~\ref{alg:CodingSpanner:FUSstep}--\ref{alg:CodingSpanner:simulation})
 is independently encoded via the Hoza-Schulman coding scheme, whose guarantees are given in  Lemma~\ref{lem:HS}. 
The first part is resilient to a $\mu_1=\Theta(1/n)$ fraction of corrupted messages and the second part is resilient to a $\mu_2=\Theta(1/s)$ fraction of corrupted messages. Since we balance  the message complexity of the two parts so that $C_1=C_2$ and since $s=\Omega(n)$, the new scheme is resilient to a fraction of $\mu_2/2 = \Theta(1/s)$ of corrupted messages overall.\footnote{\label{footnote:balance}By balancing the parts $C_1$ and $C_2$ in a weighted way  %
one can obtain a slightly improved resilience of $\mu_1\mu_2/(\mu_1+\mu_2)= \Theta(1/(n+s))$ which is, however, asymptotically equivalent to~$\Theta(1/s)$.}
\end{proof}


\paragraph{Acknowledgement.} We are grateful to Merav Parter for bringing~\cite{DMZ10} to our attention.

\bibliographystyle{alphabbrv-doi}
\bibliography{network}

\end{document}